\documentclass{jac}

\usepackage[utf8]{inputenc}
\usepackage[english] {babel} 

\usepackage{amsmath,amsfonts,amssymb}   
\usepackage{textcomp}         
\usepackage{float}
\usepackage{tikz,ifthen}
\usetikzlibrary{decorations.pathmorphing,decorations.pathreplacing,arrows,automata,patterns,positioning}
\usepackage{latexsym}

\usepackage{xspace}
\usepackage{cite}
\usepackage{url}
\usepackage{hyperref}

\date{}

\newcommand{\ZZ}{\mathbb{Z}}

\newcommand{\NN}{\mathbb{N}}
\newcommand{\QQ}{\mathbb{Q}}
\newcommand{\Sl}{\mathcal{S}}

\colorlet{myyellow}{yellow}
\colorlet{myblue}{blue!40}
\colorlet{mygray1}{black!25}
\colorlet{mygray2}{black!65}
\colorlet{colorleft}{red}
\colorlet{colorright}{yellow}

\begin{document}
\title{Slopes of Tilings}

\author
{E. Jeandel}{Emmanuel Jeandel}
\author
{P. Vanier}{Pascal Vanier}
\address
{Laboratoire d'Informatique Fondamentale de Marseille\\
CMI - 39 rue Joliot-Curie\\
F-13453 Marseille Cedex 13, France}
\email[E. Jeandel]{emmanuel.jeandel@lif.univ-mrs.fr}
\email[P. Vanier]{pascal.vanier@lif.univ-mrs.fr}
\thanks{Both authors are party supported by ANR-09-BLAN-0164.}
\begin{abstract}\noindent
We study here slopes of periodicity of tilings.
A tiling is of slope $\theta$ if it is periodic along direction
$\theta$ but has no other direction of periodicity.

We characterize in this paper the set of slopes we can achieve with tilings, and prove they coincide with recursively enumerable sets of rationals.

\end{abstract}
\keywords{tilings, formal languages, foundations of computing.}
\maketitle

\section{Introduction}\label{S:I}

The model of tilings was introduced by Wang \cite{WangII} to study fragments of
the first order theory. This model is described by geometrical local
properties, deciding whether a given tile can be placed on a given
cell based only on its surrounding neighbours.

While the definition of tilings is deceptively simple, they exhibit
complex behaviours. As an example, the most basic problem (decide if a
given tiling system can tile the plane) is undecidable \cite{Berger2}.
This is due to both a straigthforward encoding of Turing machines in
tilings \cite{buchi,Chaitin08,Emde}
and to the existence of so-called \emph{aperiodic} tiling systems
\cite{Robinson,KariR}, that can tile the plane but in no periodic way.

In this paper we explore the periodic behaviour of tiling systems.
Periodic tilings have nice closure properties, in the sense that the image
of a periodic point by a shift-preserving morphism (i.e. a block map) is
again a periodic point. As a consequence, understanding the structure of the 
periodic points of a tiling system is a first step to decide when some tiling
system embeds in another, or when two tilings systems are ``isomorphic''
(more accurately conjugate \cite{LindMarcus})

In dimension one, the question boils down to determine for a tiling system
$\tau$ the set of integers $n$ so that there is a valid tiling by $\tau$ of
period (exactly) $n$. This question was answered succesfully: Using automata
theory, a complete characterization of the set of integers we can obtain
this way was obtained\cite{LindMarcus}.

The question is more delicate in two dimensions. We might break it 
down in two parts: Given a tiling system $\tau$,
\begin{itemize}
	\item For which $n$ is there a tiling of horizontal and vertical
	  period $n$ ?
    \item For which direction $\theta$ is there a tiling which is periodic only
	  along direction~$\theta$~?
\end{itemize}
The authors gave an answer to the first question in \cite{PVEJDLT}: Sets of
integers we can obtain correspond to the complexity class \textbf{NE}.
We deal in this paper with the second question, characterizing the set of
\emph{slopes} we can obtain by tiling systems.

While the answer in dimension one involves finite automata theory, it turns
out that the good tool to solve the problem in higher dimensions is
computability theory. The undecidability of the domino problem (deciding if a
tiling system tiles the plane) is indeed not an anomaly: many combinatorial
aspects of tilings can only be fully comprehended by means of recursivity
theory arguments \cite{HochMey,AubrunS09,Meyero}.

Along these lines, we will prove here the following theorem:

\begin{theorem}
 The sets of slopes of tilings are exactly the recursively enumerable sets of rationals.
\end{theorem}
As a consequence, one might for example build a tiling system which admits
slopes arbitrary close to 0, but does not admit 0 as a slope.

This paper is organized as follows. We first give the definition of 
tiling systems, and an encoding of Turing machines that will be used later.
Then we proceed to the proof of the theorem. The main part of this paper is
a construction, for any recursively enumerable set $R$, of a tiling system
with $R$ as a set of slopes.
\section{Definitions}\label{S:Tpc}
\subsection{Tilings}

Usually when considering tiling systems, Wang rules are used.
We use here a generalization that is equivalent in terms of
expressivity  but makes the constructions easier.

While Wang rules consider only adjacent tiles only, our rules may
consider an arbitrary large (but finite) neighborhood of tiles.

A \emph{tiling} of $\ZZ^2$ with a finite set of tiles $T$ is a mapping $c:\ZZ^2\to T$.
A \emph{pattern} of neighborhood $N \subseteq \ZZ^2$ is a
mapping from $N$ to $T$. A pattern is finite if $N$ is finite.
A \emph{tiling system} is a pair $(T,F)$, where $F$ is a finite set of
finite patterns.
A tiling $c$ is said to be \emph{valid} if and only if none of the patterns of $F$ ever appear in $c$.
Since the number of forbidden patterns is finite, we could specify the rules by \emph{allowed}
patterns as well. We give an example of such a tiling system with the tiles of figure~\ref{ex1_tuiles}a
and the forbidden patterns of figure~\ref{ex1_tuiles}b. The allowed tilings are shown in figure~\ref{ex1_pavages}.

\begin{figure}[htbp]
 \begin{center}
$(a)$~\begin{tikzpicture}[scale=0.7]
 \filldraw[color=myyellow] (0,0) rectangle (1,1);
 \filldraw[color=myblue] (2,0) rectangle (3,1);
 \filldraw[color=white] (4,0) rectangle (5,1);
\end{tikzpicture}
~~~~~~~~~~~~$(b)$~\begin{tikzpicture}[scale=0.6]
 \filldraw[color=myyellow] (0,0) rectangle (1,1);
 \filldraw[color=myblue] (1,0) rectangle (2,1);
 \filldraw[color=myblue] (3,0) rectangle (4,1);
 \filldraw[color=myyellow] (4,0) rectangle (5,1);
 \filldraw[color=myyellow] (6,1) rectangle (7,2);
 \filldraw[color=myblue] (6,0) rectangle (7,1);
\end{tikzpicture}

 \end{center}
 \caption{The set of tiles $(a)$ and the forbidden patterns $(b)$.}
 \label{ex1_tuiles}
\end{figure}
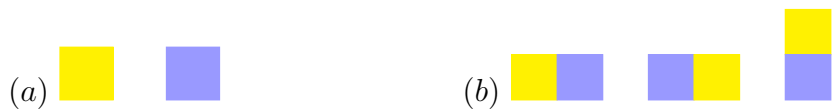

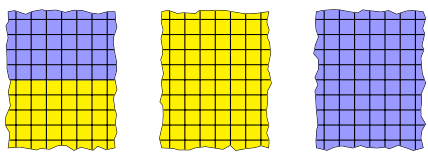
\begin{figure}[htbp]
 \begin{center}
\begin{tikzpicture}[scale=0.2]
 \clip[draw,decorate,decoration={random steps, segment length=3pt, amplitude=1pt}] (0.5,0.5) rectangle (7.5,9.5);
 \foreach \x in {0,...,7}
  \foreach \y in {0,...,4}
   \filldraw[color=myyellow] (\x,\y) rectangle (\x+1,\y+1);
 \foreach \x in {0,...,7}
  \foreach \y in {5,...,9}
   \filldraw[color=myblue] (\x,\y) rectangle (\x+1,\y+1);
 \draw (0,0) grid (8,10);
\end{tikzpicture}
~~
\begin{tikzpicture}[scale=0.2]
 \clip[draw,decorate,decoration={random steps, segment length=3pt, amplitude=1pt}] (10+0.5,0.5) rectangle (10+7.5,9.5);
 \foreach \x in {0,...,7}
  \foreach \y in {0,...,9}
   \filldraw[color=myyellow] (10+\x,\y) rectangle (10+\x+1,\y+1);
 \draw (10+0,0) grid (10+8,10);
\end{tikzpicture}
~~
\begin{tikzpicture}[scale=0.2]
 \clip[draw,decorate,decoration={random steps, segment length=3pt, amplitude=1pt}] (10+0.5,0.5) rectangle (10+7.5,9.5);
 \foreach \x in {0,...,7}
  \foreach \y in {0,...,9}
   \filldraw[color=myblue] (10+\x,\y) rectangle (10+\x+1,\y+1);
 \draw (10+0,0) grid (10+8,10);
\end{tikzpicture}

 \end{center}
 \caption{The only valid tilings of the system.}
 \label{ex1_pavages}
\end{figure}
\newpage
\subsection{(a)periodicity}
A tiling $c$ is \emph{periodic of period} $v = (v_x, v_y)\in
\ZZ^2$ if for all points $x,y\in\ZZ$, $c(x,y)=c(x+v_x,y+v_y)$.
The \emph{direction} of a vector $v \not= (0,0)$ is $\theta = v_y/v_x \in
\mathbb{Q} \cup \{\infty\}$ with the
convention $\theta = \infty$ if $v_x =0$.

A tiling is \emph{periodic along a direction $\theta$} if it is
periodic of period $v \not=(0,0)$ and $v$ is of direction
$\theta$.

For a given tiling $c$, there are three cases:
\begin{itemize}
	\item Either $c$ is periodic of period $v,w$ and $v,w$ are of
	  different directions. In this case, the tiling $c$ is \emph{biperiodic}: there exists
an integer $n\in \NN$ (the period) so that ${c(x,y) = c(x+n,y) = c(x,y+n)}$, and as a
	  consequence $c$ is periodic along all directions $\theta \in \QQ \cup \{\infty\}$
	\item $c$ is periodic along one direction $\theta$ only. In this
	  case, we will call $\theta$ the \emph{slope} of $c$.
	\item $c$ has no nonzero vector of periodicity. $c$ is then
	  called aperiodic.
\end{itemize}	
The set of slopes of a tiling system $\tau$, noted $\Sl_\tau$, is the set of the
slopes of all valid tilings by $\tau$.
As an example, the first tiling in fig.\ref{ex1_pavages} is periodic
of vector $(1,0)$ (hence of slope $0$) and the two other tilings are
biperiodic (hence have no slope). As a consequence, $\Sl_\tau = \{
  0\}$ for this example. Using rotated versions of this elementary
tiling system, we can produce for each $\theta \in \QQ \cup \{\infty\}$ a
tiling system $\tau$ so that $\Sl_\tau = \{ \theta\}$.

A tiling system is \emph{aperiodic} if and only if it tiles the plane but
all valid tilings are aperiodic. Such tiling systems have been
shown to exist \cite{Berger2} and are at the core of the undecidability of
the \emph{domino problem} (decide whether a given tiling system admits
a valid tiling).
J. Kari \cite{KariNil} gave such a tiling system with an 
interesting property: determinism.  A tiling system is
\emph{NW-deterministic} (for \emph{North-West}) if it is given by forbidden
patterns of shape \tikz[scale=0.15]{\draw (0,0) rectangle (-2,1) (0,0) rectangle (-1,2);}
and given two tiles
respectively at the north and west of a given cell, there is
at most one tile that can be put in this cell so that the finite pattern is valid.
The mechanism is shown below:

 
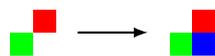
\begin{center}
\begin{tikzpicture}[scale=0.3]
 \fill[color=green] (0,0) rectangle (1,1);
 \fill[color=red] (1,1) rectangle (2,2);
 \node[right] (l) at (2,1) {};
 \begin{scope}[shift={(7,0)}]
 \node[left] (r) at (0,1) {};
 \fill[color=green] (0,0) rectangle (1,1);
 \fill[color=red] (1,1) rectangle (2,2);
 \fill[color=blue] (1,0) rectangle (2,1);
 \end{scope}
 \path[thick,-latex] (l) edge (r);
\end{tikzpicture}
 \end{center}

If we modify the forbidden patterns of this tiling system in the following way :

 
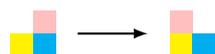
\begin{center}
\begin{tikzpicture}[scale=0.3]
 \fill[color=yellow] (0,0) rectangle (1,1);
 \fill[color=pink] (1,1) rectangle (2,2);
 \fill[color=cyan] (1,0) rectangle (2,1);
 \node[right] (l) at (2,1) {};
 \begin{scope}[shift={(7,0)}]
 \node[left] (r) at (0,1) {};
 \fill[color=yellow] (0,0) rectangle (1,1);
 \fill[color=pink] (0,1) rectangle (1,2);
 \fill[color=cyan] (1,0) rectangle (2,1);
 \end{scope}
 \path[thick,-latex] (l) edge (r);
\end{tikzpicture}
 \end{center}
 a tile will be forced by the one on its
west and on its northwest, we will call this East-determinism :

 
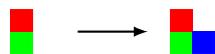
\begin{center}
\begin{tikzpicture}[scale=0.3]
 \fill[color=green] (0,0) rectangle (1,1);
 \fill[color=red] (0,1) rectangle (1,2);
 \node[right] (l) at (2,1) {};
 \begin{scope}[shift={(7,0)}]
 \node[left] (r) at (0,1) {};
 \fill[color=green] (0,0) rectangle (1,1);
 \fill[color=red] (0,1) rectangle (1,2);
 \fill[color=blue] (1,0) rectangle (2,1);
 \end{scope}
 \path[thick,-latex] (l) edge (r);
\end{tikzpicture}
 \end{center}

 East-determinism has the interesting property that if we set a whole column of the plane
 then the whole half plane on its east will be determined by it. Moreover, this tiling system
 is also aperiodic (the tilings are skewed versions of the original one; diagonal lines 
 are transformed into columns).

\subsection{Computability}

The undecidability of the domino problem \cite{Berger2} hinted earlier
also comes from a straightforward encoding of Turing machines into
tilings. We provide here such an encoding for future reference.

For a given  Turing machine $M$, consider the tiling system $\tau_M$
presented in figure~\ref{mt_tuiles}. The tiling system is given by \emph{Wang tiles}, i.e., we can
only glue two tiles together if they coincide on their common edge.
We now give some details on the picture:
\begin{itemize}
\item $s_0$ in the tiles is the initial state of the Turing machine.
\item The first tile corresponds to the case where the Turing machine, given the state $s$ and the letter $a$
chose to go to the left and to change from $s$ to $s'$, writing $a'$. The two
other tiles are similar.
\item $h$ represents a halting state. Note that the only states that can
  appear in the last step of a computation (before a border appears)
are halting states.
\end{itemize}
This tiling system $\tau_M$ has the following property: there is an accepting path
for the word $u$ in time (less than) $t$ using space (less than) $w$ if and only if we can tile a
rectangle of size $(w+2) \times t$  with white borders, the  first row containing the input.

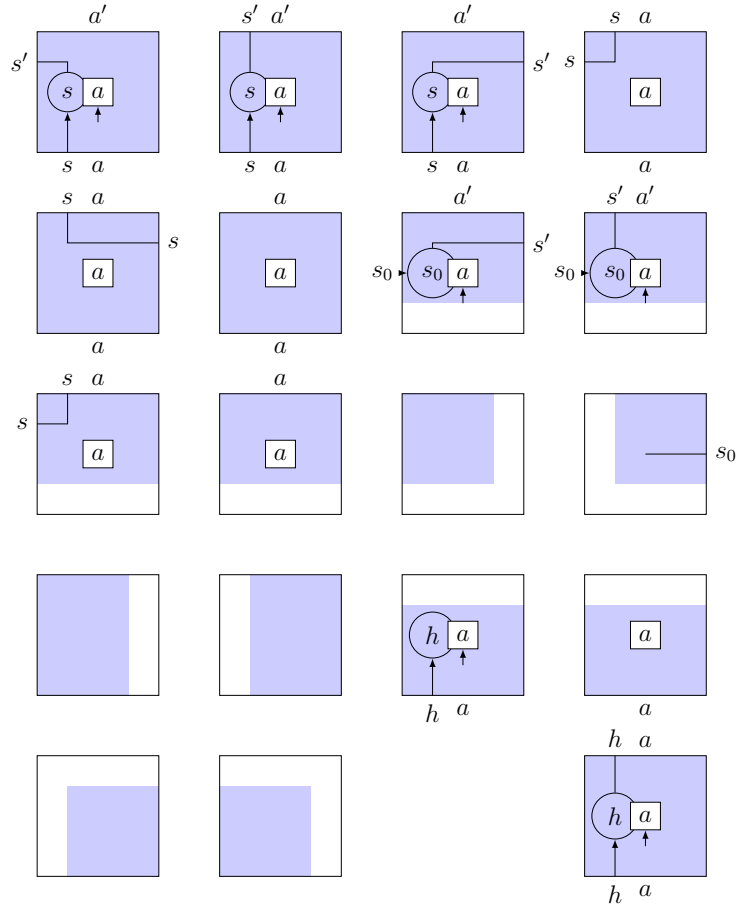
\begin{figure}
 \begin{center}
\scalebox{0.8}{
\begin{tikzpicture}[auto,scale=0.5]
\filldraw[fill=blue!20] (0,0) rectangle (4,4);
\node[draw,circle] (etat) at (1,2) {$s$};
\draw[-latex] (1,0) node[below] (basetat) {$s$} -- (etat);
\draw (etat) -- (1,3) -- (0,3);
\node[draw,rectangle,fill=white] (ruban) at (2,2) {$a$};
\draw[-latex] (2,1) -- (ruban);
\node[left] at (0,3) (hautetat) {$s'$};
\node[above] at (2,4) (hautruban) {$a'$};
\node[below] at (2,0) (basruban) {$a$};
\begin{scope}[shift={(6,0)}]
\filldraw[fill=blue!20] (0,0) rectangle (4,4);
\node[draw,circle] (etat) at (1,2) {$s$};
\draw[-latex] (1,0) node[below] (basetat) {$s$} -- (etat);
\draw (etat) -- (1,4);
\node[draw,rectangle,fill=white] (ruban) at (2,2) {$a$};
\draw[-latex] (2,1) -- (ruban);
\node[above] at (1,4) (hautetat) {$s'$};
\node[above] at (2,4) (hautruban) {$a'$};
\node[below] at (2,0) (basruban) {$a$};
\end{scope}
\begin{scope}[shift={(12,0)}]
\filldraw[fill=blue!20] (0,0) rectangle (4,4);
\node[draw,circle] (etat) at (1,2) {$s$};
\draw[-latex] (1,0) node[below] (basetat) {$s$} -- (etat);
\draw (etat) -- (1,3) -- (4,3);
\node[draw,rectangle,fill=white] (ruban) at (2,2) {$a$};
\draw[-latex] (2,1) -- (ruban);
\node[right] at (4,3) (hautetat) {$s'$};
\node[above] at (2,4) (hautruban) {$a'$};
\node[below] at (2,0) (basruban) {$a$};
\end{scope}
\begin{scope}[shift={(18,0)}]
\filldraw[fill=blue!20] (0,0) rectangle (4,4);
\draw  (0,3) -- (1,3) -- (1,4);
\node[draw,rectangle,fill=white] (ruban) at (2,2) {$a$};
\node[left] at (0,3)  (basetat) {$s$};
\node[above] at (1,4) (hautetat) {$s$};
\node[above] at (2,4) (hautruban) {$a$};
\node[below] at (2,0) (basruban) {$a$};
\end{scope}
\begin{scope}[shift={(0,-6)}]
\filldraw[fill=blue!20] (0,0) rectangle (4,4);
\draw  (4,3) -- (1,3) -- (1,4);
\node[draw,rectangle,fill=white] (ruban) at (2,2) {$a$};
\node[right] at (4,3)  (basetat) {$s$};
\node[above] at (1,4) (hautetat) {$s$};
\node[above] at (2,4) (hautruban) {$a$};
\node[below] at (2,0) (basruban) {$a$};
\end{scope}
\begin{scope}[shift={(6,-6)}]
\filldraw[fill=blue!20] (0,0) rectangle (4,4);
\node[draw,rectangle,fill=white] (ruban) at (2,2) {$a$};
\node[above] at (2,4) (hautruban) {$a$};
\node[below] at (2,0) (basruban) {$a$};
\end{scope}
\begin{scope}[shift={(12,-6)}]
\fill[color=blue!20] (0,1) rectangle (4,4);
\draw (0,0) rectangle (4,4);
\node[draw,circle] (etat) at (1,2) {$s_0$};
\draw (etat) -- (1,3) -- (4,3);
\draw[-latex] (0,2) -- (etat);
\node[draw,rectangle,fill=white] (ruban) at (2,2) {$a$};
\draw[-latex] (2,1) -- (ruban);
\node[left] at (0,2) (basetat) {$s_0$};
\node[right] at (4,3) (hautetat) {$s'$};
\node[above] at (2,4) (hautruban) {$a'$};
\end{scope}
\begin{scope}[shift={(18,-6)}]
\fill[color=blue!20] (0,1) rectangle (4,4);
\draw (0,0) rectangle (4,4);
\node[draw,circle] (etat) at (1,2) {$s_0$};
\draw (etat) -- (1,4);
\draw[-latex] (0,2) -- (etat);
\node[draw,rectangle,fill=white] (ruban) at (2,2) {$a$};
\draw[-latex] (2,1) -- (ruban);
\node[above] at (1,4) (hautetat) {$s'$};
\node[above] at (2,4) (hautruban) {$a'$};
\node[left] at (0,2) (basetat) {$s_0$};
\end{scope}
\begin{scope}[shift={(0,-12)}]
\fill[color=blue!20] (0,1) rectangle (4,4);
\draw (0,0) rectangle (4,4);
\draw (0,3) -- (1,3) -- (1,4);
\node[draw,rectangle,fill=white] (ruban) at (2,2) {$a$};
\node[above] at (1,4) (hautetat) {$s$};
\node[left] at (0,3) (basetat) {$s$};
\node[above] at (2,4) (hautruban) {$a$};
\end{scope}
\begin{scope}[shift={(6,-12)}]
\fill[color=blue!20] (0,1) rectangle (4,4);
\draw (0,0) rectangle (4,4);
\node[draw,rectangle,fill=white] (ruban) at (2,2) {$a$};
\node[above] at (2,4) (hautruban) {$a$};
\end{scope}
\begin{scope}[shift={(12,-12)}]
\fill[color=blue!20] (0,1) rectangle (3,4);
\draw (0,0) rectangle (4,4);
\end{scope}
\begin{scope}[shift={(18,-12)}]
\fill[color=blue!20] (4,1) rectangle (1,4);
\draw (0,0) rectangle (4,4);
\draw (2,2) -- (4,2);
\node[right] at (4,2) (basetat) {$s_0$};
\end{scope}
\begin{scope}[shift={(0,-18)}]
\fill[color=blue!20] (0,0) rectangle (3,4);
\draw (0,0) rectangle (4,4);
\end{scope}
\begin{scope}[shift={(6,-18)}]
\fill[color=blue!20] (4,0) rectangle (1,4);
\draw (0,0) rectangle (4,4);
\end{scope}
\begin{scope}[shift={(12,-18)}]
\fill[color=blue!20] (0,3) rectangle (4,0);
\draw (0,0) rectangle (4,4);
\node[draw,circle] (etat) at (1,2) {$h$};
\draw[-latex] (1,0) node[below] (basetat) {$h$} -- (etat);
\node[draw,rectangle,fill=white] (ruban) at (2,2) {$a$};
\draw[-latex] (2,1) -- (ruban);
\node[below] at (2,0) {$a$};
\end{scope}
\begin{scope}[shift={(18,-18)}]
\fill[color=blue!20] (0,0) rectangle (4,3);
\draw (0,0) rectangle (4,4);
\node[draw,rectangle,fill=white] (ruban) at (2,2) {$a$};
\node[below] at (2,0) (basruban) {$a$};
\end{scope}
\begin{scope}[shift={(0,-24)}]
\fill[color=blue!20] (1,0) rectangle (4,3);
\draw (0,0) rectangle (4,4);
\end{scope}
\begin{scope}[shift={(6,-24)}]
\fill[color=blue!20] (3,0) rectangle (0,3);
\draw (0,0) rectangle (4,4);
\end{scope}
\begin{scope}[shift={(18,-24)}]
\filldraw[fill=blue!20] (0,0) rectangle (4,4);
\node[draw,circle] (etat) at (1,2) {$h$};
\draw[-latex] (1,0) node[below] (basetat) {$h$} -- (etat);
\draw (etat) -- (1,4);
\node[draw,rectangle,fill=white] (ruban) at (2,2) {$a$};
\draw[-latex] (2,1) -- (ruban);
\node[above] at (1,4) (hautetat) {$h$};
\node[above] at (2,4) (hautruban) {$a$};
\node[below] at (2,0) (basruban) {$a$};
\end{scope}
\end{tikzpicture}
}
 
 \end{center}
 \caption{A tiling system, given by Wang tiles, simulating a Turing machine : the states 
 are in the circles and the tape is in the rectangles.}
 \label{mt_tuiles}
\end{figure}


\section{The sets of slopes are recursively enumerable}\label{S:Ssre}
We say that a subset $S$ of $\QQ \cup \{ \infty\}$ is \emph{recursively
enumerable} if there exists a Turing machine $M$ that on input
$(p,q) \in \mathbb{Z}^2 \not= (0,0)$ halts if and only if $q/p \in S$.

\[ \begin{array}{rcl}
	\theta \in S &\implies& \forall (p,q),   q/p = \theta, M \text{
	  halts on }(p,q) \\
\theta \not\in S &\implies& \forall (p,q),   q/p = \theta, M \text{
	  does not halt on }(p,q) 
\end{array}\]

The exact definition is irrelevant as all reasonable definitions will
give rise to the same class.
An alternative interesting definition is as follows: A set $S$ is
recursively enumerable if there exists a Turing machine $M$ so that
\[ \theta \in S \iff \exists (p,q),   q/p = \theta \wedge M \text{
	  halts on }(p,q) \]

Using a known projection technique to go down to dimension 1, 
we prove here:

\begin{lemma}
For any tiling system $\tau$, $S_\tau$ is recursively enumerable.
\end{lemma}
\begin{proof}
We first give a procedure to decide if there is a tiling which is
$(n,0)$-periodic.
Let $k$ be an integer bigger than the size of any forbidden pattern in
$\tau$.

If $w$ is a pattern of support $[0,n-1] \times [0,l]$ for some $l$,
we write $w^\ZZ$ for the pattern of support $\ZZ \times [0,l]$ defined
by $w^\ZZ_{i,j} = w_{(i\mod n),j}$, that is for the horizontal
repetition of $w$.

Let $V$ be the set of all patterns $w$ of size $n \times k$ so that $w^\ZZ$ is correctly tiled.
Consider this a directed graph $G$, where there is an edge from $v$ to $w$
if and only if $(v \otimes w)^\ZZ$ is correctly tiled, where
$v \otimes w$ denotes the pattern of size $n \times 2k$ obtained by
putting $w$ above $v$.

It is then clear that tilings of period $(n,0)$ correspond to biinfinite
walks on this graph, so that there exists a tiling of period $(n,0)$ if 
and only if there exists a cycle in the graph $G$. Furthermore, there
exist a tiling of period $(n,0)$ which is not biperiodic if and only if
we can find two distinct cycles $C_1, C_2$ in the graph so that $C_2$
is accessible from $C_1$. All the construction is clearly algorithmic.

Now for a given $(p,q)$ we use the same procedure, where $w$ is a
pattern of size $|p| \times k|q|$ and
$w^\ZZ$ is of support $\{ (i+np, j+nq), i \leq |p|, j \leq k|q|\}$
and defined by $w^\ZZ_{i+np,j+nq}= w_{i ,j}$.

The following algorithm gives then the expected result: Starting from a given
$(p,q)$, test all  $(p',q')$ so that $q'/p' = q/p$ to see if there
exists a tiling which is $(p',q')$-periodic but not biperiodic.
\end{proof}
\newpage
\section{The recursively enumerable sets are sets of slopes}\label{S:Ress}

\begin{lemma}
For any recursively enumerable set $R\subseteq\QQ\cup\{\infty\}$, there exists a tiling system $\tau$, such
that $\Sl_\tau=R$.
\end{lemma}
\proof
We use for this proof techniques similar to \cite{PVEJDLT}.
We will construct for each Turing machine $M$, corresponding to a recursively
enumerable set $R$, a tiling system $\tau$ whose
slopes are exactly the rationals $\theta$ accepted by $M$.
We assume that $M$ takes $\theta$ as an input under the form $(p,q)$
in binary and that its input depends only on $q/p$.

We will first build a tiling system $\tau$ that has as slopes $\{\theta  \in   R | 0 < \theta < 1\}$.
The other cases are treated in the same way and the final tiling system is the disjoint union of
the tiling systems treating each case. The special cases $\theta=0$,
$\theta=\infty$, and $\theta=\pm 1$ will be shortly discussed later on.

For the particular case where $p>q>0$ we want to enforce
the fact that when a tiling of the plane has exactly one direction of
periodicity, this direction of periodicity has to be accepted by the Turing 
machine $M$. The tiling $\tau_M$ will enforce the skeleton described
in figure~\ref{skel}, where each square encodes the computation by $M$
proving that the slope $\theta$ is accepted. For this, we need the size of the square to be
arbitrarily large independently of $\theta$, so that the computation of $M$ has enough time
to accept. This skeleton in itself could be biperiodic, 
we will then color the background of each square
to ensure the existence of tilings with only one direction of periodicity.

\begin{figure}[!h]
 \begin{center}
  \begin{tikzpicture}[scale=0.35,xscale=1.58]  
  \draw[clip,decorate,decoration={random steps,segment length=3pt,amplitude=1pt}] (0.5,0) rectangle (14.5,11);
  \def\a{3};
  \foreach \i in {-1,0,...,4}{
   \draw (\i*\a,0) -- (\i*\a,20);
    \foreach \j in {0,...,2}
    {
     \draw (\i*\a,\i+\j*1.5*\a) -- (\i*\a+\a,\i+\j*\a*1.5);
    }
    \foreach \j in {-1,...,2}
    {
     \draw[thick,-latex] (\i*\a+\a/2,\i+\j*\a*1.5+\a*0.72) -- (\i*\a+\a*1.5,\i+\j*\a*1.5+1.05*\a);
    }
  }
  \end{tikzpicture}
 \end{center}
 \caption{Skeleton of the tiling : when the tiling is periodic, the squares appear and each of 
 them is the shifted version of its lower left neighbor. Inside the squares we will encode the 
 Turing machine.}
 \label{skel}
\end{figure}
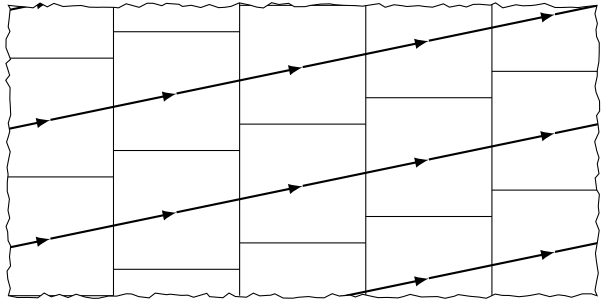

In order to enforce this skeleton, we will use several layers (or components), each of them having 
their own aim, and impose some contraints on how the layers may combine.
We give here $\tau_M=C\times R\times W\times S\times P\times T_M\times A$ where :
\begin{itemize}
 \item $C$ will allow us to make the rows and columns,
 \item $R$ to make the squares,
 \item $W$ to force the periodicity vector and to write the input for the Turing machine,
 \item $S$ to force the aperiodic background of the squares to be the same,
 \item $P$ will reduce the size of the input,
 \item $T_M$ will code the Turing machine $M$,
 \item $A$ will allow slopes of unique periodicity to appear.
\end{itemize}
We will now proceed to the details of the proof, by giving each component and explaining what
it enforces.
\newcommand{\tnoir}{\tikz[scale=0.25]{
\fill[color=black] (0,0) rectangle (1,1);
}\xspace}

\newcommand{\rightmost}{\scalebox{0.25}{\tikz{
\draw (1,1) rectangle (0,0);
\filldraw[fill=black] (0,0) -- (0.5,0.5) -- (0,1) -- cycle;
\filldraw[fill=mygray2] (0,0) -- (0.5,0.5) -- (1,0) --cycle;
\filldraw[fill=mygray1] (0,1) -- (0.5,0.5) -- (1,1) -- cycle;
}}\xspace}

\newcommand{\leftmost}{\scalebox{0.25}{\tikz{
\draw (1,1) rectangle (0,0);
\filldraw[fill=black] (1,0) -- (0.5,0.5) -- (1,1) -- cycle;
\filldraw[fill=mygray2] (0,1) -- (0.5,0.5) -- (1,1) -- cycle;
\filldraw[fill=mygray1] (0,0) -- (0.5,0.5) -- (1,0) --cycle;
}}\xspace}

\newcommand{\betweenrl}{\tikz[scale=0.25]{\filldraw[color=black,fill=mygray2] (1,1) rectangle (0,0);}\xspace}

\newcommand{\betweenlr}{\tikz[scale=0.25]{\filldraw[color=black,fill=mygray1] (1,1) rectangle (0,0);}\xspace}

\begin{description}
\item[Component $C$] The first component is made of an \textbf{East-deterministic} aperiodic set of tiles that we will 
 call white tiles (the white background of figure~\ref{skel}), and we add two sets of tiles
 the horizontal breaking tiles $\{\tnoir \}$ and the vertical breaking tiles $\left\{\leftmost,
 \rightmost,\betweenrl,\betweenlr\right\}$ (the horizontal and vertical lines of figure~\ref{skel}). 
 The rules are simple :
 \begin{itemize}
  \item on the left of a \tnoir there can only be a \tnoir or a \leftmost,
  \item on the right of a \tnoir there can only be a \tnoir or a \rightmost,
  \item above and below a \tnoir, there can only be a white,
  \item above a \leftmost can only be a \betweenrl,
  \item above a \betweenrl can only be a \rightmost or a \betweenrl,
  \item above a \rightmost can only be a \betweenlr,
  \item above a \betweenlr can only be a \leftmost or a \betweenlr.
 \end{itemize}

 To put it in a nutshell, it means that horizontal breaking tiles forms rows that can only
 be broken by vertical breaking tiles, and vertical breaking tiles can only form columns
 that cannot be broken.
 
 In a periodic tiling, we cannot have a quarter of plane filled with
 white (aperiodic tiles). As a consequence, periodic tilings at
 this stage are necessarily formed by a white
 background broken \emph{infinitely many times} by horizontal or vertical
 breaking tiles.

 One more rule we add is that the rules on white tiles
 "jump" over the black tiles. That is to say if we remove a black row, then the white tiles have
 to glue themselves together correctly.
 The valid tilings at this stage are represented on 
 figure~\ref{afterC}.

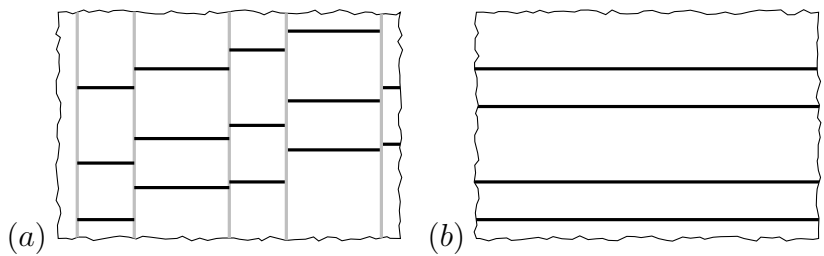
\begin{figure}
 \begin{center}
  \begin{tikzpicture}[scale=0.25]
  \node[left] at (-1,1) {$(a)$};
  \begin{scope}
\draw[clip,decorate,decoration={random steps,segment length=3pt,amplitude=1pt}] (-1,1) rectangle +(18,12);
   \draw[very thick,color=mygray1] (0,0) -- (0,20) (3,0) -- (3,20) (8,0) -- (8,20);
   \draw[very thick,shift={(8,2)}] (0,2) -- (3,2) (0,5) -- (3,5) (0,9) -- (3,9);
   \draw[very thick,shift={(16,4)}] (0,2) -- (3,2) (0,5) -- (3,5) (0,9) -- (3,9);
   \draw[very thick] (0,2) -- (3,2) (0,5) -- (3,5) (0,9) -- (3,9);
  \draw[very thick] (3,6.3) -- (8,6.3) (3,3.7) -- (8,3.7) (3,10) -- (8,10);
  \draw[very thick,shift={(8,2)}] (3,6.3) -- (8,6.3) (3,3.7) -- (8,3.7) (3,10) -- (8,10);
   \draw[very thick,color=mygray1] (11,0) -- (11,20) (16,0) -- (16,20) (19,0) -- (19,20);
  \end{scope}
  \node[left] at (21,1) {$(b)$};
  \begin{scope}[shift={(20,2)}]
\draw[clip,decorate,decoration={random steps,segment length=3pt,amplitude=1pt}] (1,-1) rectangle +(18,12);
   \draw[very thick] (0,0) -- (20,0) (0,2) -- (20,2) (0,6) -- (20,6) (0,8) -- (20,8);
  \end{scope}
  \end{tikzpicture}
 \end{center}
 \caption{Valid periodic tilings are formed of columns of vertical breaking tiles $(a)$ or 
 of rows of horizontal breaking tiles $(b)$. Between two columns of vertical breaking tiles there
 can be rows of horizontal breaking tiles.}
 \label{afterC}
\end{figure}

\newcommand{\tvert}{\scalebox{0.25}{\tikz{
\fill[fill=colorright] (0.5,0) rectangle (1,1);
\fill[fill=colorleft] (0,0) rectangle (0.5,1);
\draw (0,0) rectangle (1,1);
}}\xspace}

\newcommand{\thoriz}{\scalebox{0.25}{\tikz{
\fill[fill=colorleft] (0,0) rectangle (1,0.5);
\fill[fill=colorright] (0,0.5) rectangle (1,1);
\draw (0,0) rectangle (1,1);
}}\xspace}

\newcommand{\tdiag}{\scalebox{0.25}{\tikz{
\fill[fill=colorright] (0,0) -- (0,1) -- (1,0) -- cycle;
\fill[fill=colorleft] (0,1) -- (1,0) -- (1,1) -- cycle;
\draw (0,0) rectangle (1,1);
}}\xspace}

\newcommand{\tgauche}{\scalebox{0.25}{\tikz{
\filldraw[fill=colorleft] (0,0) rectangle (1,1);
}}\xspace}

\newcommand{\tdroite}{\scalebox{0.25}{\tikz{
\filldraw[fill=colorright] (0,0) rectangle (1,1);
}}\xspace}

\newcommand{\jointureg}{\scalebox{0.25}{\tikz{
\filldraw[fill=colorright] (0,0) rectangle (1,1);
\fill[fill=colorleft] (0,0) rectangle (0.5,1) (0.5,0.5) -- (1,0) -- (1,0.5) --cycle;
\draw (0,0) rectangle (1,1);
}}\xspace}

\newcommand{\jointured}{\scalebox{0.25}{\tikz{
\filldraw[fill=colorleft] (0,0) rectangle (1,1);
\fill[fill=colorright] (0.5,0) rectangle (1,1);
\fill[fill=colorright] (0,1) -- (0.5,0.5) -- (0,0.5) --cycle;
\draw (0,0) rectangle (1,1);
}}\xspace}
\item[Component $R$]. The next component will force the apparition of squares between
 two columns of vertical breaking tiles and prevent several infinite rows
 of horizontal breaking tiles to appear. This layer is made of the
 set of tiles $\left\{\tvert,\thoriz,\tdiag,\tgauche,\tdroite,\jointured,\jointureg
 \right\}$, the rules applied
 on this layer are given by Wang tiles. We superimpose the rules as follows :
 \begin{itemize}
  \item \tvert can only be superimposed to \betweenrl,\betweenlr,
  \item \thoriz can only be superimposed to \tnoir,
  \item \jointureg goes on \leftmost, and \jointured goes on \rightmost,
  \item \tgauche,\tdroite,\tdiag are superimposed to the white tiles.
 \end{itemize}

Figure~\ref{compoR} shows how this component $R$ forces rows of black 
 tiles to appear between two gray columns. The distance
 between these black rows is exactly the distance between the gray columns thus black rows and gray
 columns form squares. At this stage the valid periodic tilings cannot be formed of only rows 
 of black tiles anymore. 
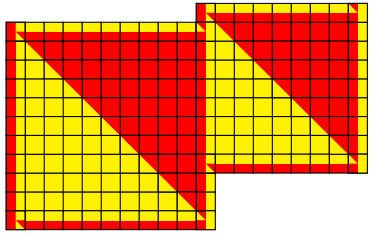
\begin{figure}
 \begin{center}
  \begin{tikzpicture}[scale=0.25]
  \begin{scope}
   \foreach \i in {1,...,9}{
    \foreach \j in {1,...,\i}{
\filldraw[fill=colorright] (9-\i,\j) rectangle +(1,1);
\filldraw[fill=colorleft] (\i,10-\j)  rectangle +(1,1);
    }
   }
   \foreach \i in {1,...,9}{
\fill[fill=colorright] (0,\i) +(0.5,0) rectangle +(1,1);
\fill[fill=colorleft] (0,\i) +(0,0) rectangle +(0.5,1);
\draw (0,\i) +(0,0) rectangle +(1,1);
\fill[fill=colorleft] (\i,0) +(0,0) rectangle +(1,0.5);
\fill[fill=colorright] (\i,0) +(0,0.5) rectangle +(1,1);
\draw (\i,0) +(0,0) rectangle +(1,1);
\fill[fill=colorright] (\i,10-\i) +(0,0) -- +(0,1) -- +(1,0) -- cycle;
\fill[fill=colorleft] (\i,10-\i) +(0,1) -- +(1,0) -- +(1,1) -- cycle;
\draw (\i,10-\i) +(0,0) rectangle +(1,1);
\fill[fill=colorright] (10,\i) +(0.5,0) rectangle +(1,1);
\fill[fill=colorleft] (10,\i) +(0,0) rectangle +(0.5,1);
\draw (10,\i) +(0,0) rectangle +(1,1);
\fill[fill=colorleft] (\i,10) +(0,0) rectangle +(1,0.5);
\fill[fill=colorright] (\i,10) +(0,0.5) rectangle +(1,1);
\draw (\i,10) +(0,0) rectangle +(1,1);
   }
   \foreach \p in {(0,0),(0,10)}{
\filldraw[fill=colorright] \p +(0,0) rectangle +(1,1);
\fill[fill=colorleft] \p +(0,0) rectangle +(0.5,1) +(0.5,0.5) -- +(1,0) -- +(1,0.5) --cycle;
\draw \p +(0,0) rectangle +(1,1);
}
   \foreach \p in {(10,0),(10,10)}{
\filldraw[fill=colorleft] \p +(0,0) rectangle +(1,1);
\fill[fill=colorright] \p +(0.5,0) rectangle +(1,1);
\fill[fill=colorright] \p +(0,1) -- +(0.5,0.5) -- +(0,0.5) --cycle;
\draw \p +(0,0) rectangle +(1,1);
  }
  \end{scope}
  \begin{scope}[shift={(10,3)}]
   \foreach \i in {1,...,7}{
    \foreach \j in {1,...,\i}{
\filldraw[fill=colorright] (7-\i,\j) rectangle +(1,1);
\filldraw[fill=colorleft] (\i,8-\j)  rectangle +(1,1);
    }
   }
   \foreach \i in {1,...,7}{
\fill[fill=colorright] (0,\i) +(0.5,0) rectangle +(1,1);
\fill[fill=colorleft] (0,\i) +(0,0) rectangle +(0.5,1);
\draw (0,\i) +(0,0) rectangle +(1,1);
\fill[fill=colorleft] (\i,0) +(0,0) rectangle +(1,0.5);
\fill[fill=colorright] (\i,0) +(0,0.5) rectangle +(1,1);
\draw (\i,0) +(0,0) rectangle +(1,1);
\fill[fill=colorright] (\i,8-\i) +(0,0) -- +(0,1) -- +(1,0) -- cycle;
\fill[fill=colorleft] (\i,8-\i) +(0,1) -- +(1,0) -- +(1,1) -- cycle;
\draw (\i,8-\i) +(0,0) rectangle +(1,1);
\fill[fill=colorright] (8,\i) +(0.5,0) rectangle +(1,1);
\fill[fill=colorleft] (8,\i) +(0,0) rectangle +(0.5,1);
\draw (8,\i) +(0,0) rectangle +(1,1);
\fill[fill=colorleft] (\i,8) +(0,0) rectangle +(1,0.5);
\fill[fill=colorright] (\i,8) +(0,0.5) rectangle +(1,1);
\draw (\i,8) +(0,0) rectangle +(1,1);
   }
   \foreach \p in {(0,0),(0,8)}{
\filldraw[fill=colorright] \p +(0,0) rectangle +(1,1);
\fill[fill=colorleft] \p +(0,0) rectangle +(0.5,1) +(0.5,0.5) -- +(1,0) -- +(1,0.5) --cycle;
\draw \p +(0,0) rectangle +(1,1);
}
   \foreach \p in {(8,0),(8,8),(0,7)}{
\filldraw[fill=colorleft] \p +(0,0) rectangle +(1,1);
\fill[fill=colorright] \p +(0.5,0) rectangle +(1,1);
\fill[fill=colorright] \p +(0,1) -- +(0.5,0.5) -- +(0,0.5) --cycle;
\draw \p +(0,0) rectangle +(1,1);
  }
  \end{scope}
  \end{tikzpicture}
 \end{center}
 \caption{Component $R$ forces squares.}
 \label{compoR}
\end{figure}

\item[Component $W$] What this component does is that it synchronises the
offsets between squares of two neighboring columns, and forces all columns to be at equal distance
of their two neighboring columns, for all of them. As a side
effect, it also writes the offset between two squares (which we call
$q$) in each square.
In order to do that, what we do is that we prolongate the black rows of
each column into their direct neighbors with two new layers, one for the left and one for the right.
The end of the black row then sends a diagonal signal which changes its direction when it collides
with the projected lines of the neighbors and its colision with the column has to coincide with the 
projection of the other column. Figure~\ref{projections}.a shows how this mechanism works.
The collision of the signal sent on the right extremity of the black lines
marks the end of the input $q$ on each square.
We add two other sublayers to make the white rows of same width. The first one sends a signal from
the left extremity of a black line which has to meet the next column at the exact point 
of the extension
of the square. The second one does the same for the right extremity. 
Figure~\ref{projections}.b shows these signals.

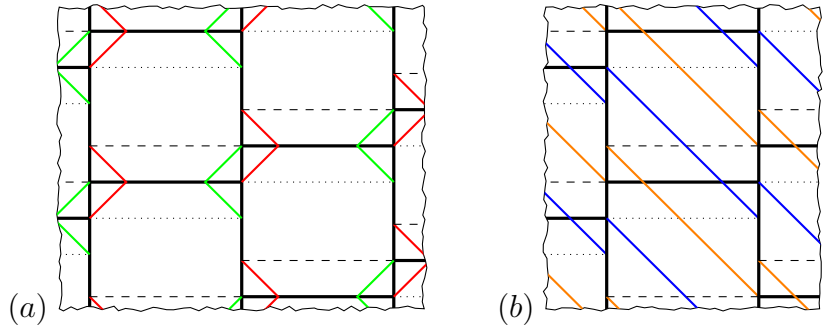
\begin{figure}
 \begin{center}
  \begin{tikzpicture}[scale=0.4]
\begin{scope}
\node[left] at (4,2) {$(a)$};
\draw[clip,decorate,decoration={random steps,segment length=3pt,amplitude=1pt}] (4,2) rectangle +(12,10);
  \def\taille{5};
  \def\decalage{1.2};
  \foreach \i in {-4,...,6}{
   \draw[very thick] (\i*\taille,0) -- (\i*\taille,20);
    \foreach \j in {-6,...,10}
    {
\draw[very thick] (\i*\taille,\i*\decalage+\j*\taille) -- (\i*\taille+\taille,\i*\decalage +\j*\taille);
\draw[dotted] (\i*\taille+\taille,\i*\decalage+\j*\taille) -- (\i*\taille+2*\taille,\i*\decalage +\j*\taille);
\draw[dashed] (\i*\taille,\i*\decalage+\j*\taille) -- (\i*\taille-\taille,\i*\decalage +\j*\taille);
\draw[thick,color=red] (\i*\taille,\i*\decalage+\decalage+\j*\taille) -- (\i*\taille+\decalage,\i*\decalage +\j*\taille) -- (\i*\taille,\i*\decalage-\decalage+\j*\taille);
\draw[thick,color=green] (\i*\taille,\i*\decalage+\j*\taille) -- (\i*\taille-\decalage,\i*\decalage-\decalage +\j*\taille) -- (\i*\taille,\i*\decalage-2*\decalage+\j*\taille);
    }
  }
\end{scope}
\begin{scope}[xshift=17cm]
\node[left] at (3,2) {$(b)$};
\draw[clip,decorate,decoration={random steps,segment length=3pt,amplitude=1pt}] (3,2) rectangle +(9,10);
  \def\taille{5};
  \def\decalage{1.2};
  \foreach \i in {-4,...,6}{
   \draw[very thick] (\i*\taille,0) -- (\i*\taille,20);
    \foreach \j in {-6,...,10}
    {
\draw[very thick] (\i*\taille,\i*\decalage+\j*\taille) -- (\i*\taille+\taille,\i*\decalage +\j*\taille);
\draw[dotted] (\i*\taille+\taille,\i*\decalage+\j*\taille) -- (\i*\taille+2*\taille,\i*\decalage +\j*\taille);
\draw[dashed] (\i*\taille,\i*\decalage+\j*\taille) -- (\i*\taille-\taille,\i*\decalage +\j*\taille);
\draw[thick,color=blue] (\i*\taille,\i*\decalage -\decalage +\j*\taille) -- (\i*\taille + \taille,\i*\decalage - \decalage +\j*\taille-\taille);
\draw[thick,color=orange] (\i*\taille,\i*\decalage +\decalage +\j*\taille) -- (\i*\taille + \taille,\i*\decalage + \decalage +\j*\taille-\taille);
    }
  }
\end{scope}
  \end{tikzpicture}
 \end{center}
 \caption{
 The dotted row (resp. dashed) corresponds to the prolongation on 
 the right (resp. left) of the black cells.
 In $(a)$ the signals sent from the extremities of the rows forming the square forces 
 the offset between rectangles of three neighboring columns to be 
 exactly the same for any of them. In $(b)$ the signals sent from the extremities 
 force the distance between columns to be identical.
 }
 \label{projections}
\end{figure}

\newcommand{\trsdroite}{\scalebox{0.25}{\tikz{
\draw (0,0) rectangle (1,1);
\draw[line width=3pt,-latex] (0,0.5) -- (1,0.5);
}}\xspace}
\newcommand{\trsdroiteg}{\scalebox{0.25}{\tikz{
\filldraw[fill=gray] (0,0) rectangle (1,1);
\draw[line width=3pt,-latex] (0,0.5) -- (1,0.5);
}}\xspace}

\newcommand{\trsdiagg}{\scalebox{0.25}{\tikz{
\filldraw[fill=gray] (0,0) rectangle (1,1);
\draw[line width=3pt,-latex] (0,0) -- (1,1);
}}\xspace}

\newcommand{\trsdiag}{\scalebox{0.25}{\tikz{
\draw (0,0) rectangle (1,1);
\draw[line width=3pt,-latex] (0,0) -- (1,1);
}}\xspace}

\item[Component $S$] This component is meant to synchronize the 
aperiodic backgrounds of all the squares. In order to do that, we only need to
transmit the first column after a vertical breaking column
since our initial aperiodic tiling system is East-deterministic. 

In order to do that, we
take these tiles $\{\trsdroite,\trsdiag,\trsdiagg,\trsdroiteg\}$, with the following rules :
\begin{itemize}
 \item on the right,above and below  a \trsdiag there can only be a \trsdiag or a \trsdiagg. 
 \item on the left of a \leftmost we necessarily have a \trsdiagg and the south 
  western neighbor of a \trsdiagg, if the tile is a white, is a \trsdiagg or a \trsdroiteg,
 \item the lower left white tile of a square is necessarily a \trsdroiteg. The rules 
  on \trsdroiteg is that there can only be a \trsdroiteg or a \trsdiag on a white tile to its right,
 \item the vertical/horizontal breaking tiles have necessarily a \trsdroite on them. 
\end{itemize}
The tiling obtained inside a square is shown on figure~\ref{transmi}.
We add a sublayer that is a copy of the white tiles with the rules that the tiles of this
component on the right of this column are identical to the white ones on component $C$ and
that this copy is transmitted to the tile pointed by the arrow. Then with the property that 
the black tiles continue the rules on the whites, the whole aperiodic background between
two vertical breaking columns is exactly the same but shifted by the offset.

\begin{figure}
 \begin{center}
\scalebox{0.25}{
  \begin{tikzpicture}
\newcommand{\tn}{
\fill[color=black] (0,0) rectangle (1,1);
}

\newcommand{\rmost}[2]{
\draw (#1,#2) +(0,0) rectangle +(1,1);
\filldraw[fill=black] (#1,#2) +(0,0) -- +(0.5,0.5) -- +(0,1) -- cycle;
\filldraw[fill=mygray2] (#1,#2) +(0,0) -- +(0.5,0.5) -- +(1,0) --cycle;
\filldraw[fill=mygray1] (#1,#2) +(0,1) -- +(0.5,0.5) -- +(1,1) -- cycle;
}
\newcommand{\lmost}[2]{
\draw (#1,#2) +(0,0) rectangle +(1,1);
\filldraw[fill=black] (#1,#2)  +(1,0) -- +(0.5,0.5) -- +(1,1) -- cycle;
\filldraw[fill=mygray2] (#1,#2) +(0,1) -- +(0.5,0.5) -- +(1,1) -- cycle;
\filldraw[fill=mygray1] (#1,#2) +(0,0) -- +(0.5,0.5) -- +(1,0) --cycle;
}

\foreach \i in {1,...,9}{
    \foreach \j in {0,...,6}{
	\draw (\j,\i) rectangle +(1,1);
	\draw[line width=3pt,-latex] (\j,\i) +(0,0.5) -- +(1,0.5);
    }}
   \foreach \i in {0,...,6}{
	\filldraw[fill=gray] (\i,0) rectangle ++(1,1);
	\draw[line width=3pt,-latex] (\i,0) +(0,0.5) -- +(1,0.5);
    }
   \foreach \i in {0,...,9}{
    \foreach \j in {1,...,3}{
	\ifthenelse{\i=\j}{
	\filldraw[fill=gray] (\j+6,\i) rectangle +(1,1);
	\draw[line width=3pt,-latex] (\j+6,\i) -- +(1,1);
	}{
	\draw (\j+6,\i) rectangle +(1,1);
	\draw[line width=3pt,-latex] (\j+6,\i) -- +(1,1);
	}
    }}
 \fill[fill=mygray2] (-1,0) rectangle (0,6); 
 \fill[fill=mygray1] (-1,7) rectangle (0,10); 
 \fill[fill=mygray2] (10,4) rectangle (11,10); 
 \fill[fill=mygray1] (10,0) rectangle (11,3); 
 \filldraw[fill=black] (0,0) rectangle +(10,-1);
 \filldraw[fill=black] (0,10) rectangle +(10,1);
 \rmost{10}{10}
 \rmost{10}{-1}
 \rmost{-1}{6}
 \lmost{10}{3}
 \lmost{-1}{-1}
 \lmost{-1}{10}
 \draw (-1,-1) grid (0,11);
 \draw (10,0) grid (11,11);
\end{tikzpicture}
 }
 \end{center}
 \caption{Tiles allowing to transmit the aperiodic background.}
 \label{transmi} 
\end{figure}
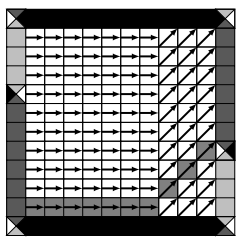

\item[Component $P$]
Now each square contains two data: its size ($p$) and the offset to
the next square $q$, both in unary. We will pass them as input to 
the Turing machine after some transformation.

The idea is to transform the unary input $(p,q)$ into a
smaller binary one $(p',q')$ where $\gcd(p',q')$ is not a multiple of two. Doing that is fairly easy :
we first need to convert the input in binary; this 
can be done by the iteration of the transducer  of figure~\ref{transdu_conv}: 
starting from $000\dots 00$ we obtain the binary representation of $p$ (least significant bit on the rightmost part) in $p$
iterations of the transducer. Then we strip the
binary representation of $p$ and $q$ of their common last zeroes.

\begin{figure}
\begin{center}
\begin{tikzpicture}[->,node distance=3cm,auto,initial text=]
        
        \tikzstyle{every state}=[minimum size=10mm]

\node[state,accepting] (A)                    {$q_1$};
\node[initial,state]         (C) [left of=A] {$q_0$};


\path (A) edge [loop above]   node[above] {$1|0$} (A);
\path (C) edge   node[above] {$0|1$} (A);
\path (C) edge [loop above]   node[above] {$1|1,0|0$} (C);

\end{tikzpicture}
\end{center}
\caption{A transducer tranforming $n$ in binary into $n+1$.}
\label{transdu_conv}
\end{figure}
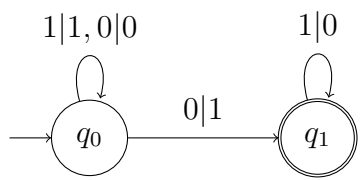

\item[Component $T_M$] This layer implements the Turing machine $M$, the 
input has been computed by layer $P$. Note that the Turing machine has
to halt for the tiling to be valid.

\item[Component $A$] This layer is made of only two tiles, a yellow and a blue one.
It will be superimposed to white tiles and to the \betweenrl of the vertical breaking
tiles of component $C$ only.
The rules are that two neighboring tiles (horizontally and vertically) have the 
same color. It is easy to see that the color is uniform inside a square and that 
it spreads to the upper right and lower left neighboring squares. Thus the squares along the
direction of periodicity have the same color.
\end{description}

We now prove that the preceding construction works.
\begin{description}
\item[(1) Any slope is an accepted input of $M$]

Let $\theta = q/p \in S_\tau$ be a slope of periodicity of $\tau$, 
with $p>q>0$ relatively prime.

By construction, the tiling has to be formed of squares of identical
size with constant offset (components $C$, $R$, $W$). Their aperiodic background
has to be the same on each column (component $S$), so that in fact the
tiling is periodic along direction $(m,n)$ where $m$ and $n$ denote
respectively the width and offset of the tiling. As a consequence, the
tiling is of slope $\theta=m/n=q/p \in ]0;1[$ and  we have $(n,m)=2^kk'(p,q)$ for some $k,k'$ with $k'$ odd.

Now the Turing Machine on each square has $(k'q,k'p)$
as an input and halts. Hence the slope $k'q/k'p$ is accepted by the
machine, so $q/p \in R$, which proves $S_\tau \subseteq R\ \cap\ ]0 ;1 [$.

\item[(2)  Any accepted input of $M$ is a slope of some tiling]

Let $\theta \in R$ be an accepted input of $M$ with
$\theta = q/p$, $p>q>0$ and $p,q$ relatively prime.

There exists a time $t$ and a space $s$ such that $M$ accepts $(p,q)$
in time $t$ and space $s$ and $s\leq t$.
Take $(m,n) = 2^{\left\lceil \log t\right\rceil} (p,q) \geq (t,s)$
Now the $m \times m$ square is big enough for the computation on input
$(p,q)$ to succeed. Hence there is a tiling of period
$(m,n)$ and component $A$ allows us to make the direction of periodicity
unique by dividing the plane into two colors, half a plane yellow and half a
plane blue. Hence $R\ \cap\ ]0 ;1 [ \subseteq S_\tau$. 
\end{description}
This finishes the proof for the case $0 <\theta < 1$, i.e. $p > q > 0$.

The cases where
$q>p>0$, $-p>q>0$, or $q>-p>0$ are treated in a very similar way: rotating the tiling system
we just constructed and changing the way the input is written on the tape (to invert the inputs,
or add a minus sign) is enough. However the remaining cases ($p=\pm q,
p=0, q=0$) need special treatment\footnote{As this corresponds to four
  specific different $\theta$s, note that we could treat them
  nonconstructively, adding if necessary four new tiling systems having
  predescribed slopes. 
  }.

For these cases, the 
construction above does not work, by that we mean that just rotating it and modifying
slightly the Turing machine of component $T_M$ won't do the trick. 
However it is actually simpler. We now make squares facing
one another, obtaining a regular grid. This requires less tiles for
component $C$ and no component $W$.
Then according to the case, components $C$,$S$ and $A$ are modified as
follows:
\begin{itemize}
 \item for $p=q$ ($\theta=1$), $S$ just transmits diagonally the tiles. In
 component $A$, the color is synchronized from the top right corner to the next square at
 the north east. The case $p= -q$ is similar.
 \item for $q=0$ ($\theta=0$), $S$ transmits horizontally, and the colors of component $A$ are synchronized
 with the square on the right. The tiling can only be horizontally periodic if the Turing machine
 accepts it, this is the only way it can be periodic.
 \item for $p=0$ ($\theta=\infty)$, $C$ has, instead of an east deterministic tileset, a north deterministic one.
 Components $S$ and $A$ are modified accordingly. The tiling can only be vertically periodic if
 the Turing machine accepts it and this is the only way it can be periodic.
\qed\end{itemize}

\section{Concluding remarks}
We have shown that the sets of slopes of periodicity of tilings correspond exactly to the 
recursively enumerable ($\Sigma^0_1$) sets of rationals for tilings in dimension 2. Our intuition for
analogous results in higher dimensions would be that the slopes of periodicity would then
be characterized by $\Sigma^0_2$ sets\cite{odifreddi}, since knowing whether a tiling is
periodic of vector $v$ in dimension $3$ is not decidable anymore but only
$\Pi^0_1$. Hence the following conjecture:

\begin{conjecture}
 The sets of slopes of tilings in dimension $d\geq 3$ are exactly the $\Sigma^0_2$ subsets
 of $(\QQ\cup\{\infty\})^{d-1}$.
\end{conjecture}
An analogous construction to the one detailed here should work at least for dimension 3,
it would however be tedious.

\bibliographystyle{plain}
\bibliography{biblio,books}{}

\end{document}